\documentclass[11pt,a4paper]{article}
\usepackage[titletoc,title]{appendix}
\usepackage{graphicx}
\usepackage{color}
\usepackage{subcaption}
\usepackage{setspace}
\usepackage[backend=bibtex,style=numeric,natbib=true]{biblatex}
\usepackage{amsmath,amsfonts,amssymb,bm,hyperref,amsthm,mathrsfs}
\numberwithin{equation}{section}
\newtheorem{theorem}{Theorem}[section]
\newtheorem{lemma}[theorem]{Lemma}

\newtheorem{corollary}[theorem]{Corollary}
\theoremstyle{definition}

\newtheorem{remark}[theorem]{Remark}
\newtheorem{definition}[theorem]{Definition}

\usepackage[top=1in, bottom=1in, left=1in, right=1in]{geometry}
\usepackage{fancyhdr}
\begin{document}

\baselineskip=2\baselineskip

\title{Minimizing minor embedding energy: an application in quantum annealing}
\author{Yan-Long Fang\thanks{YF:
London Centre for Nanotechnology,
University College London,
19 Gordon Street,
London WC1H~0AH,
UK;
yanlong.fang@ucl.ac.uk.
}
\and
P.A. Warburton\thanks{PAW:
London Centre for Nanotechnology,
University College London,
19 Gordon Street,
London WC1H~0AH,
UK;
p.warburton@ucl.ac.uk.
}}

\renewcommand\footnotemark{}


\date{\today}

\maketitle
\begin{abstract}
A significant challenge in quantum annealing is to map a real-world problem onto a hardware graph of limited connectivity. If the maximum degree of the problem graph exceeds the maximum degree of the hardware graph, one employs minor embedding in which each logical qubit is mapped to a tree of physical qubits. Pairwise interactions between physical qubits in the tree are set to be ferromagnetic with some coupling strength $F<0$. Here we address the question of what value $F$ should take in order to maximise the probability that the annealer finds the correct ground-state of an Ising problem. The sum of $|F|$ for each logical qubit is defined as minor embedding energy. We confirm experimentally that the ground-state probability is maximised when the minor embedding energy is minimised, subject to the constraint that no domain walls appear in every tree of physical qubits associated with each embedded logical qubit. We further develop an analytical lower bound on $|F|$ which satisfies this constraint and show that it is a tighter bound than that previously derived by Choi (Quantum Inf. Proc. 7 193 (2008)). 

%

\

{\bf Keywords: }
Minor embedding; adiabatic quantum computing; job-shop scheduling.

\

\end{abstract}

\newpage
\tableofcontents

\section{Introduction}
\label{Introduction}
Quantum annealing is a widely-used tool for solving quadratic optimization problems \cite{Harris2018,King2018}. The problem is mapped to a Hamiltonian, $H_P$, whose ground-state encodes the optimized solution. Exploration of the potential landscape is driven by quantum fluctuations described by a driver Hamiltonian, $H_D$. The overall system Hamiltonian $H_{total}$ is a time-varying weighted sum of $H_P$ and $H_D$ such that at the end of the annealing process the quantum fluctuations are suppressed and $H_{total} = H_P$. A typical annealing schedule is of the form
\begin{equation}
\label{adiabatic H}
H_{total}(s)=A(s) H_D+B(s) H_P\;,
\end{equation}
where $0\le s\equiv \frac{t}{t_f}\le 1$, $t$ is time, $t_f$ is the duration of the anneal, $A(0)\gg B(0)$ and $A(1)\ll B(1)$. The origin of quantum annealing goes back to the quantum adiabatic theorem with a gap condition, which was first shown by Born and Fock \cite{BornFock1928} in 1928, then Kato \cite{Kato1950} simplified the proof of the theorem and extended it to allow degenerate eigenstates and eigenvalue crossings. For closed quantum systems, Farhi et al. \cite{Farhi2000,Farhi2001} proposed adiabatic quantum computation as an alternative to tackle NP-complete problems. For a recent review of the quantum adiabatic theorem, see for exmaple Albash and Lidar \cite{Lidar2018}.


In view of the computational complexity of modelling interacting quantum systems using classical computational resources, a potentially efficient way to find the ground-state of $H_P$ is to engineer a physical system whose dynamics follow that of equation \eqref{adiabatic H}. One such physical system is based on a system of superconducting flux qubits with tunable inductive interactions\cite{Kafri2017}. In this implementation the problem Hamiltonian is of the Ising form:
\begin{equation}
\label{adiabatic Ising HP}
H_P=\sum_i h_i \sigma_i^z +\sum_{ij\in E(G)}J_{ij}\sigma_i^z \sigma_j^z \;.
\end{equation}
Here $\sigma_i^z$ is the quasi-spin of qubit $i$ (corresponding to its flux state) and $G$ is a graph describing all possible two-qubit interactions. The total Hamiltonian is exactly the transverse Ising model introduced by Kadowaki and Nishimori\cite{Nishimori1998}, which is a quantum analogue of classical simulated annealing. Moreover, many NP-hard problems can be translated into Ising Hamiltonians \cite{Lucas2014}. Now the expression \eqref{adiabatic H} becomes 
\begin{equation}
\label{adiabatic Ising}
H_{total}(s)=A(s) \sum_i \tilde{h}_i \sigma_i^x+B(s) \left(\sum_i h_i \sigma_i^z +\sum_{ij\in E(G)}J_{ij}\sigma_i^z \sigma_j^z \right)\;.
\end{equation}

One problem for hardware implementation of quantum annealing now becomes immediately apparent: for a system of $N$ qubits it is at best very difficult to engineer direct interactions between all $\frac{1}{2}N(N-1)$ pairs. In current implementations of flux-qubit quantum annealers the maximum degree of the hardware graph is 6 -- i.e. each qubit is directly coupled to at most six other qubits\footnote{experiments are currently underway on a flux-qubit annealer with degree $15$}. It is therefore necessary to employ \textit{minor embedding} -- i.e. to embed an Ising problem Hamiltonian whose connectivity graph has degree $D_P$ onto physical hardware with connectivity graph of degree $D_H$, where $D_H < D_P \le N$. The requirement of this embedding is that the ground-state of the embedded Hamiltonian of degree $D_H$ encodes the same solution as the ground-state of the problem Hamiltonian of degree $D_P$. 

Choi \cite{Choi2008} first proposed a method for minor embedding in which each logical qubit is replaced by a tree of physical qubits. All the physical qubits within each tree are constrained to be in the same spin state (which in turn is the spin state of the logical qubit) by the implementation of ferromagnetic interactions of magnituede $|F|$ at each edge of the tree. In practice it is usual to use a one-dimensional chain of physical qubits as the tree for minor embedding. A logical qubit consisting of a chain of $L$ physical qubits in a hardware graph of degree $D_H$ can now be directly coupled to $L (D_H - 2) + 2$ other logical qubits, thereby greatly increasing the connectivity. Figure \ref{minorfig} shows an example of a minor embedding.
\begin{figure}[h]
	\centering
	\includegraphics[scale=0.4]{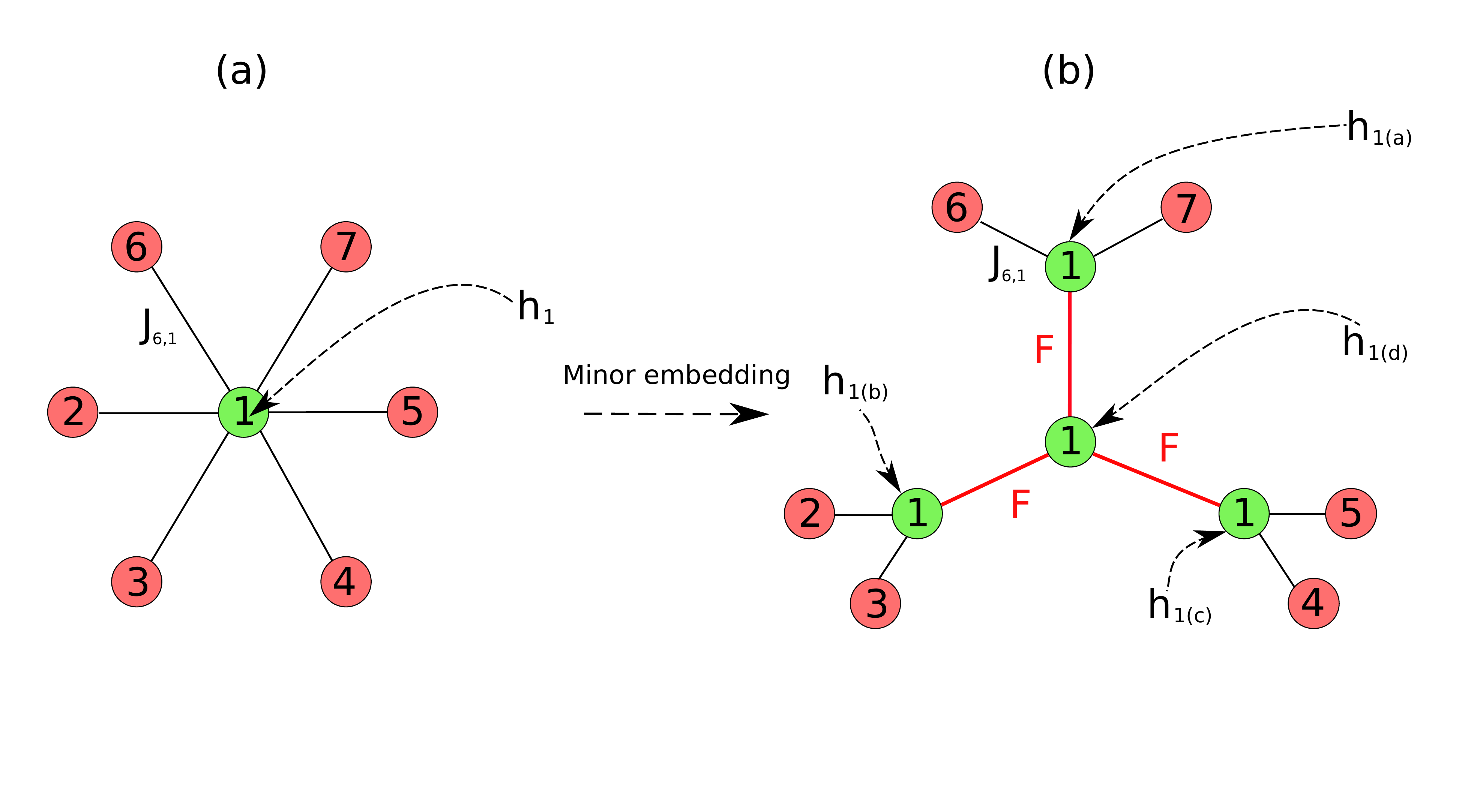}
	\caption{ An illustrative example of (a) a logical graph of maximum degree $6$ and (b) a physical graph of maximum degree $3$. Logical qubit $1$ (coloured in green in (a)) is mapped onto four physical qubits (all labelled by $1$ and coloured in green in (b)). $J_{6,1}$ in (a) denotes the coupling between the sixth logical qubit and the first qubit, which is mapped identically onto (b). $h_1$ in (a) is the local field on the first logical qubit, which is mapped onto $h_{1(a)}, h_{1(b)}, h_{1(c)} \& h_{1(d)}$ in (b). Other couplers and local fields are omitted for clarity.}
	\label{minorfig}
\end{figure}

If $|F|$ is sufficiently large, for a closed-system quantum annealer it can be assumed that the ferromagnetic bonds between each physical qubit in the embedded logical qubit are never broken, ensuring that all the physical spins are mutually aligned. In a real quantum annealer, however, thermal fluctuations and other noise mechanisms may break ferromagnetic bonds resulting in domain walls between locally aligned regions. In this case the value of the logical spin cannot be unambiguously determined (although majority vote may be used to estimate it). In such a real quantum annealer therefore the probability that the embedded Hamiltonian anneals to the correct ground-state depends upon the probability of domain walls forming, which in turn is a function of the strength, $F$, of the ferromagnetic interaction between the physical qubits in the embedded tree. While at first sight it might appear that the ground-state probability is monotonic in $F$, in a real quantum annealer the maximum absolute coupling strength between any pair of physical qubits is finite. (In a flux qubit annealer, for example, this maximum coupling is determined by the magnitudes of the persistent current and mutual inductances.) Arbitrary increases in the embedding ferromagnetic coupling strength normalized with respect to the energy scale of the problem Hamiltonian can therefore only be achieved by reducing the latter. This in turn leads to an increase in computational errors from thermal transitions to an excited state. Furthermore, if $F$ is too small, domain walls will be present unavoidably. This suggests that there is an optimum value for the embedding ferromagnetic coupling strength for any given embedding of the problem Hamiltonian. See Appendix \ref{Job-shop scheduling problems on the D-Wave Machine} for experimental confirmation of this supposition.
\

Several strategies for parameter setting on quantum annealers are developed by Pudenz in \cite{Pudenz} to understand how the ferromagnetic coupling strength (within embedded chains) would affect the probability of finding ground-states on the D-Wave DW2 and DW2X machine. Pudenz's work focuses on mixed satisfiability problems. It shows that higher ferromagnetic coupling strengths do not increase the chance of finding the ground-state on either machine. Moreover, different strategies for setting the logical field magnitude $h_{i(k)}$ within the chains yield different performance. In particular, the so-called single distribution method is less effective than other methods. This is due to the fact that non-admissible minor embeddings are more likely to be used in the single distribution method -- see Remark \ref{rmk on broken chains} below for details. Venturelli et.al \cite{Venturelli2015} studied the Sherrington-Kirkpatrick Model (SKM) on the D-Wave DW2 machine. They experimentally confirmed the non-monotonic dependence of the ground-state probability on $F$ by using the D-Wave quantum annealer for up to $N = 30$ fully-connected logical spins.

In this paper we revisit minor embedding in order to determine the optimum ferromagnetic strength $|F|$ for embedding trees in quantum annealers at finite temperature. We will give a mathematical criterion for the best bound on the value of $|F|$. As a consequence, the first two theorems by Choi \cite{Choi2008} will follow immediately. It is not hard to see that Choi's first paper in minor embedding \cite{Choi2008} gives the foundation for the Chimera architecture of D-Wave machines given in \cite{Choi2011}. Moreover, methods to generate minor embedding on the Chimera graph can be found in \cite{Boothby2015}. Therefore, we focus here on the analysis of minor embeddings rather than on architectures of quantum annealers. Moreover, we will see in Subsection \ref{Admissible minor embeddings} that condition \eqref{embedding condition 1.1} will influence the bound of $|F|$. Our results can be applied to any architecture as long as the Ising nature is preserved. Here the Ising nature should be understood in the broad sense. i.e. including higher order interaction terms. It is known that Hamiltonians with higher-order interactions can be reproduced via a two-body Hamiltonian (see e.g. \cite{Warburton2017}). In order to achieve multi-body interactions via two-body Ising models, one has to couple logical qubits with ancilla qubits, which certainly increases the (vertex) degree of the corresponding two-body Hamiltonian. Minor embedding is the key tool to convert graphs with higher degrees to graphs with lower degrees. Therefore, our paper will also be useful for generating multi-body interactions.

It still remains open to model the open system effectively. A simplified version can be found in \cite{Lidar2012}, where a system-bath Hamiltonian is studied in detail. The Hamiltonian is given by
$$
H(t)=H_S(t)+H_B+H_I\,,
$$
where $H_S$, $H_B$ and $H_I$ correspond to the adiabatic system, bath and interaction Hamiltonians respectively. Note that $H_I=g\sum A_\alpha\otimes B_\alpha$. This special feature enable us to use a perturbative method for small $g$ as shown in the paper \cite{Lidar2012}. However, if $g$ depends non-trivially on the strong coupling, $F$, introduced by $H_S(t)$, then $g$ might become large for large $F$. Consequently, small order perturbations will not be enough to analyse the behaviour of the system. Therefore, if $H_S=H_{total}$ and we want to use the model in \cite{Lidar2012}, we need to minimise the strong coupling, $F$, in $H_{total}$ without destroying the Ising problem $H_P$ in \eqref{adiabatic H}. This give another motivation for us to search for the minimum coupling strength in $H_P$.

\section{Main results}
\label{Main results}
\subsection{Preparatory material}
\label{Preparatory material}
Firstly, we give a formal definition for minor embedding.

\begin{definition}
	\label{def of minor}
	A minor-embedding\cite{Choi2008} is a pair of mappings $(\iota,\tau)=:I$ that maps a graph $G$ to a sub-graph of another graph $U$. The pair of mappings satisfies the following properties:
	\begin{itemize}
		\item $\iota: V(G) \mapsto V(U)$ each vertex $i$ in $V(G)$ is mapped to a set of vertices (denoted by$\iota(i)$) of a connected sub-tree of $U$\,,
		\item $\tau: V(G) \times V(G) \mapsto V(U)$ such that for each $ij\in E(G)$, $\tau(i,j)\in \iota(i)$ and $\tau(j,i)\in \iota(j)$ fulfilling $\tau(i,j)\tau(j,i)\in E(U)$. Note that $\tau$ induces the mapping of edges, which we also denote by $\tau$\,.
	\end{itemize} 	
\end{definition}
Note that given graphs $G$ and $U$, there may be no minor embedding of $G$ into $U$ or there may exist many $(\iota,\tau)$'s that embed $G$ into $U$. For instance, by Kuratowski’s theorem the complete bipartite graph $K_{3,3}$ cannot be minor embedded into any planar graph. Figure \ref{minorfig} illustrates how to embed a highly connected graph into a less connected graph.

Let $G$ be the logical graph corresponding to expression \eqref{adiabatic Ising HP}. To show its dependence on $G$, we suppress the subscript $P$ and rewrite the expression as
\begin{equation}
\label{adiabatic Ising G}
H_G=\sum_{i\in V(G)} h_i \sigma_i^z +\sum_{ij\in E(G)}J_{ij}\sigma_i^z \sigma_j^z \;.
\end{equation}
Suppose that there is another graph $U$, which we can interpret as the hardware graph. Moreover, we assume that graph $G$ can be minor embedded onto graph $U$. Then Definition \ref{def of minor} induces a series of problem Hamiltonians associated with graph $I(G)\subset U$:
\begin{equation}
\label{adiabatic Ising U}
H_{I(G)}=\sum_{i\in V(G)} \left(\sum_{k\in V(\iota(i))}h_{i(k)} \sigma_{i(k)}^z 
+
\sum_{i_p i_q\in E(\iota(i))} F^{pq}_i \sigma_{i_p}^z\sigma_{i_q}^z\right) 
+
\sum_{ij \in E(G)}J_{ij}\sigma_{\tau(i,j)}^z \sigma_{\tau(j,i)}^z \;,
\end{equation}
where
\begin{equation*}
\label{embedding condition 1}
\sum_{k\in V(\iota(i))}h_{i(k)} = h^\prime_i\,,
\end{equation*}
and the ferromagnetic coupling strength (also called internal coupling strength) within each sub-tree $\iota(i)$ is bounded from above.
\begin{equation}
\label{embedding condition 2}
F^{pq}_i < -M_i\,, \qquad \text{for some non-negative $M_i$}\,.
\end{equation}
In order to match the ground-state of Hamiltonian \eqref{adiabatic Ising G} and that of Hamiltonian \eqref{adiabatic Ising U}, we can set $h^\prime_i = h_i$, which gives
\begin{equation}
\label{embedding condition 1.1}
\sum_{k\in V(\iota(i))}h_{i(k)}= h_i\,.
\end{equation}
We also require that $M_i$ be sufficiently large that all spins in the ground-state of the embedded tree are aligned.

\

A natural question to ask is:
\emph{\textbf{How small can $M_i$ be?}}

Let $\mathscr{E}_G$ be the energy corresponding to Hamiltonian \eqref{adiabatic Ising G} and $\mathscr{E}_{I(G)}$ for Hamiltonian \eqref{adiabatic Ising U}. Then we have 
\begin{equation}
\label{adiabatic Ising energy G}
\mathscr{E}_G(s_1,\dots,s_N)=\sum_{i\in V(G)} h_i s_i +\sum_{ij\in E(G)}J_{ij} s_i s_j \;,
\end{equation}
and
\begin{multline}
\label{adiabatic Ising energy U}
\mathscr{E}_{I(G)}\left(s_{1(1)},\dots,s_{1(|\iota(1)|)},\dots,s_{N(|\iota(N)|)}\right)
\\
=\sum_{i\in V(G)} \left(\sum_{k\in V(\iota(i))}h_{i(k)} s_{i(k)} 
+
\sum_{i_p i_q\in E(\iota(i))} F^{pq}_i s_{i_p} s_{i_q}\right) 
+
\sum_{ij \in E(G)}J_{ij} s_{\tau(i,j)} s_{\tau(j,i)} \;.
\end{multline}

\begin{definition}[Minor embedding energy]
Let $I=(\iota,\tau)$ be a minor embedding. Then its minor embedding energy (MEE) is defined by
$$
\mathscr{E}_{I_{MEE}}:=\sum_{i_p i_q\in E(\iota(i))} |F^{pq}_i|\,.
$$
\end{definition}
\noindent Note that minimizing $M_i$ for each logical qubit $i$ is equivalent to minimizing the minor embedding energy.

\subsection{Main theorem}
Our task is to find the mathematical criteria for all the bounds that preserve the ground-state configuration of Hamilton \eqref{adiabatic Ising G}. Now we will focus on the criteria for tree $\iota(i)$. 
\begin{definition}[Boundary operator]
	Let $X$ be a graph and $2^{X}$ denote the power set of $V(X)$. The boundary operator 
	$$\partial : 2^{X} \mapsto E(X)$$
	is defined as that for any $W \subset V(X)$, $\partial W$ gives the boundary edges of $W$. That is the cut(s) between $W$ and $X\backslash W$. Moreover, the boundary operator $\partial$ annihilates both the empty set and the total set $V(X)$.
\end{definition}
We will see later that the boundary operator has a strong relationship with the ferromagnetic coupling strength. For a graph with assignments (local $h$-field) on each vertex, we define the following integral operator.
\begin{definition}[$h$-integral operator]
	\label{h integral operator}
	Let $X$ be a graph. The $h$-integral operator
	\[
	h: V(X) \mapsto \mathbb{R}
	\]
	is defined as
	\[
	h(W)=\sum_{k\in V(W)} h_k\qquad\qquad\text{for any $W \subset X$}\,.
	\]
\end{definition}
Similarly, we can define the $J$-integral operator for other non-negative external field.
\begin{definition}[$J$-integral operator]
	Let $X$ be a graph. The $J$-integral operator
	\[
	J: V(X) \mapsto \mathbb{R}_+
	\]
	is defined as
	\[
	J(W)=\sum_{k\in V(W)} J_k\qquad\qquad\text{for any $W \subset X$}\,.
	\]
\end{definition}
At least one domain wall is present when there is the presence of an inhomogeneous spin configuration in $\iota(i)$ or equivalently the presence of an anisotropic magnetization. 
\begin{definition}[Domain wall]
If all particles have the same spin in $W_i\subset \iota(i)$ but opposite spin in $\iota(i)\backslash W_i$, then $\partial W_i$ is the domain wall associated with $W_i$.
\end{definition}
\noindent We say a domain wall $\partial W_i$ is positive (negative), if the spins are positive (negative) within $W_i$. 

Let us denote $\operatorname{Onbh}(i(k))$ the original neighbourhood of the pre-embedded vertex $i$ that is connected to the embedded vertex $i(k)$.

Now we are ready to state our main theorem.
\begin{theorem}
	\label{main thm}
	Let $h_{i(k)}$ be the local fields and $J_{i(k)}:=\sum_{l\in \operatorname{Onbh}(i(k))}|J_{l,i(k)}|$ be the non-negative external fields on $\iota(i)$. Let $M_i$ be the constant defined in \eqref{embedding condition 2} satisfying
	\begin{equation}
	\label{main thm eqn 1}
	M_i \ge \underset{W_i}{\operatorname{max}}\left( \frac{1}{|\partial W_i|}\operatorname{min}\Big\{|h(W_i)-J(W_i)|,|h(W_i)-h_i-J(\iota(i)\backslash W_i)| \Big\}\right)\,,
	\end{equation}
	where the maximum is taken from all $\emptyset\ne W_i \subsetneq \iota(i)$. Then we have 
	\begin{equation}
	\label{main thm eqn 2}
	s^*_{i_p}s^*_{i_q}=1\,, \qquad \text{for all $i_p i_q\in E(\iota(G))$}\,,
	\end{equation}
	and 
	\begin{equation}
	\label{main thm eqn 3}
	\operatorname{min}\mathscr{E}_{I(G)}\left(s^*_{1(1)},\dots,s^*_{1(|\iota(1)|)},\dots,s^*_{N(|\iota(N)|)}\right)=\mathscr{E}_G(s^*_1,\dots,s^*_N)\,,
	\end{equation}
	where $s^*_k=s^*_{k(j)}$, for all $j\in \iota(k)$.
\end{theorem}

\begin{remark}\label{main thm rmk}
\

\begin{itemize}
	\item If certain conditions are satisfied, then the bound given in inequality \eqref{main thm eqn 1} is valid for the worst-case scenario i.e. it takes into account all possible spin configurations in the neighbourhood of the logical qubit. See Subsection \ref{Tightness of the bound} for details.
	\item It gives the necessary condition such that $M_i$ will preserve the equivalence of ground-states for $\mathscr{E}_{I(G)}$ and $\mathscr{E}_G$. Moreover it is the necessary condition for the $h_{i(k)}$'s and $J_{i(k)}$'s being pre-defined. Hence $M_i$ depends on $h_{i(k)}$ and $J_{i(k)}$. In practice, the $J_{i(k)}$'s are defined for a given minor embedding. However, the $h_{i(k)}$'s need to be determined. Therefore, the true optimal $M_i$ should be
	\begin{equation*}
	M_i=\min_{h_{i(k)}} M_i\left(h_{i(k)}\right)\,,
	\end{equation*}
	provided that some conditions are satisfied, see Section \ref{Tightness of the bound}.
\end{itemize}
\end{remark}
\noindent We will see later how this will give the true optimal bound for a simple example. Now we show that two important theorems of minor embedding by Choi \cite{Choi2008} follow as corollaries of our main theorem.
\begin{corollary}[Choi's first theorem]\label{Choi thm 1}
Let $M_i$ be the constant defined in \eqref{embedding condition 2} satisfying
\begin{equation}
\label{Choi thm 1 eqn 1}
M_i \ge |h_i|+\sum_{j\in \operatorname{nbh}(i)}|J_{ij}|\,,
\end{equation}
where $\operatorname{nbh}(i)$ means the neighbourhood of vertex $i$. We have 
\begin{equation}
\label{Choi thm 1 eqn 2}
s^*_{i_p}s^*_{i_q}=1\,, \qquad \text{for all $i_p i_q\in E(\iota(G))$}\,,
\end{equation}
and 
\begin{equation}
\label{Choi thm 1 eqn 3}
\operatorname{min}\mathscr{E}_{I(G)}\left(s^*_{1(1)},\dots,s^*_{1(|\iota(1)|)},\dots,s^*_{N(|\iota(N)|)}\right)=\mathscr{E}_G(s^*_1,\dots,s^*_N)\,,
\end{equation}
where $s^*_k=s^*_{k(j)}$, for all $j\in \iota(k)$.
\end{corollary}
\begin{proof}
It suffices to show that 
\begin{equation}
\label{corallary 1 eqn 1}
|h_i|+\sum_{j\in \operatorname{nbh}(i)}|J_{ij}|
\ge
\operatorname{max}_{W_i\subset \iota(i)}\left( \frac{1}{|\partial W_i|}\operatorname{min}\Big\{|h(W_i)-J(W_i)|,|h(W_i)-h_i-J(\iota(i)\backslash W_i)| \Big\}\right)\,.
\end{equation}
Since for each $W_i\subset\iota(i)$, we have
\[
|h_i|+\sum_{j\in \operatorname{nbh}(i)}|J_{ij}|
\ge
|h(W_i)-J(W_i)|
\ge
\frac{1}{|\partial W_i|} |h(W_i)-J(W_i)|\,,
\]
the inequality \eqref{corallary 1 eqn 1} follows immediately.
\end{proof}
In order to get Choi's tighter bound for the ferromagnetic coupler strengths, one needs to introduce the following object.
\begin{equation}
\label{def C}
C(i):=\sum_{j\in\operatorname{nbh}(i)}|J_{ij}|-|h_i|\,, \qquad \text{for all $i \in V(G)$}\,,
\end{equation}
which defines whether the spin of particle $i$ is locally determinable or non-determinable. When $C(i)<0$, the spin of particle $i$ is locally determinable, as the local field $h_i$ is dominant, whereas when $C(i) \ge 0$, its spin must be determined globally. Without loss of generality, we can assume $C(i) \ge 0$. Now, we are ready to state our second corollary.
\begin{corollary}[Choi's second theorem]\label{Choi thm 2}
Let $h_{i(k)}$ satisfy
\begin{equation}
\label{Choi thm 2 eqn 1}
h_{i(k)}=\operatorname{sgn}(h_i)\left\{
\begin{aligned}
&\sum_{\tau(j,i)\in\operatorname{Onhb}(i(k))}|J_{ij}|-\frac{C(i)}{l(i)}  \,, \qquad& &\text{where $i(k)$ is one of the $l(i)$ leaves of $\iota(i)$}\,;\\
&\sum_{\tau(j,i)\in\operatorname{Onhb}(i(k))}|J_{ij}| & &\text{otherwise}\,,
\end{aligned}
\right.
\end{equation}
where $\operatorname{Onbh}(i(k))$ means the original neighbourhood of vertex $i(k)\in \iota(i)$.Then
\begin{equation}
\label{Choi thm 2 eqn 2}
M \ge \frac{l(i)-1}{l(i)}C(i) \qquad \text{for all $i \in V(G)$}
\end{equation}
yields the same result as Corollary \ref{Choi thm 1}.
\end{corollary}
\begin{remark}[Comparison between Choi's two theorems]

\

\begin{itemize}
	\item Corollary \ref{Choi thm 1} is independent of the values of the $C(i)$'s and is certainly larger than the bound given in Corollary \ref{Choi thm 2}. However, Corollary \ref{Choi thm 1} does not assign any value to $h_{i(k)}$, whereas Corollary \ref{Choi thm 2} holds only when the $h_{i(k)}$'s satisfy equations \eqref{Choi thm 2 eqn 1}.
	\item Corollary \ref{Choi thm 2} gives the best bound when $C(i)=0$ for all $i\in V(G)$.
	\item The larger (weaker) bound given by Corollary \ref{Choi thm 1} does not require any topological information about the minor embedding, while the smaller (stronger) bound given by Corollary \ref{Choi thm 2} depends non-trivially on the topology of the minor embedding.
	\item Both proofs for Corollary \ref{Choi thm 1} and Corollary \ref{Choi thm 2} are quite different and there is no obvious derivation from Corollary \ref{Choi thm 1} to Corollary \ref{Choi thm 2}.
\end{itemize}
\end{remark}
Now we give a simple proof of Choi's second theorem as a corollary.
\begin{proof}
It suffices to show that 
\begin{equation}
\label{corallary 2 eqn 1}
\frac{l(i)-1}{l(i)}C(i)
\ge
\operatorname{max}_{W_i\subset \iota(i)}\left( \frac{1}{|\partial W_i|}\operatorname{min}\Big\{|h(W_i)-J(W_i)|,|h(W_i)-h_i-J(\iota(i)\backslash W_i)| \Big\}\right)\,,
\end{equation}
for $h_{i(k)}$ setting as in equations \eqref{Choi thm 2 eqn 1} and for all $\emptyset\ne W_i\subsetneq \iota(i)$. Now we have
\[
|h(W_i)-J(W_i)|
=
\left|\partial (L(i)\cap W_i)\right|\times \frac{C(i)}{l(i)} \,,
\]
where $L(i)$ is the set of leaves in $\iota(i)$. As $|\partial W_i|\ge 1$ for $\emptyset\ne W_i\subsetneq \iota(i)$, one can easily verify that
\[
\left|\partial (L(i)\cap W_i)\right|\le |\partial W_i|\left(|\partial L(i)|-1\right)=|\partial W_i|\left(l(i)-1\right)\,.
\] 
Therefore, we have
\[
\frac{1}{|\partial W_i|}|h(W_i)-J(W_i)|\le \frac{l(i)-1}{l(i)} C(i)\,,
\]
for all $\emptyset\ne W_i\subsetneq \iota(i)$, which completes the proof.
\end{proof}
As it remains open on the tightness of the bound in Corollary \ref{Choi thm 2}, we will give a simple example in the next subsection, which shows that even for $h_{i(k)}$'s given as in equation \eqref{Choi thm 2 eqn 1}, the bound is not tight. Furthermore, by relaxing the condition \eqref{Choi thm 2 eqn 1}, one can achieve the best bound.

\subsection{An example: existence of a tighter bound}\label{An example}
\begin{figure}[h]
	\centering
	\includegraphics[scale=0.3]{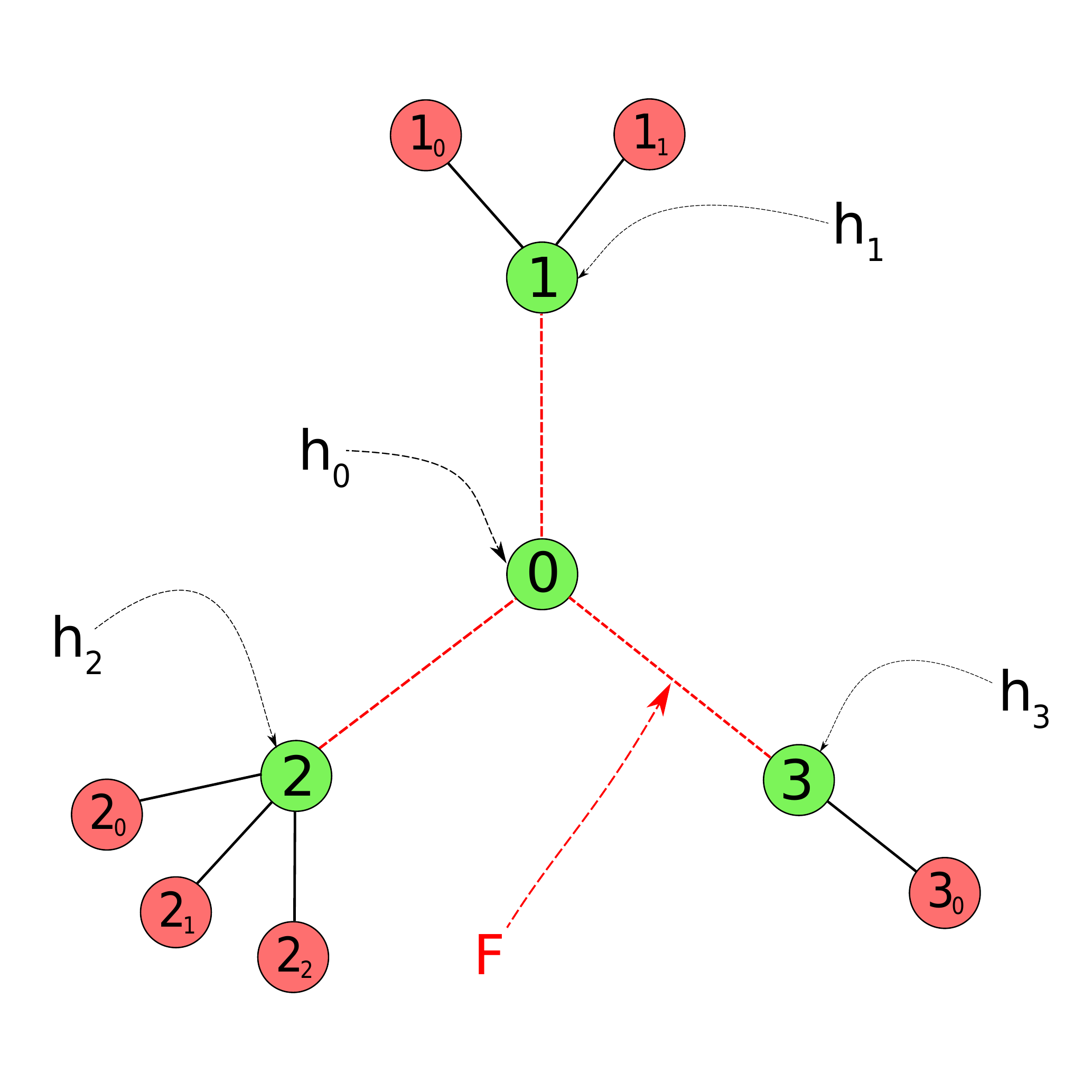}
	\caption{An example of $\iota(i)$. The green tree represents the minor embedding of $i$-th logical qubit, where the local field $h_i$ has been split into $h_1$, $h_2$, $h_3$ and $h_4$.}
	\label{minorexmaple1}
\end{figure}
In this subsection, we give an example to show the existence of a tighter bound for the ferromagnetic coupling strength compared with Corollary \ref{Choi thm 2}. Let us consider the minor embedding of a vertex $i$ as in Figure \ref{minorexmaple1}. For the sake of this example we set the couplers and local fields such that
\begin{equation}
\sum_{1(k)\in\operatorname{Onbh}(1)}|J_{1,1(k)}|=\sum_{2(k)\in\operatorname{Onbh}(2)}|J_{2,2(k)}|=\sum_{3(k)\in\operatorname{Onbh}(3)}|J_{3,3(k)}|=5h>0\,,
\end{equation}
and
\begin{equation}
h_i=3h\,.
\end{equation}
According to Corollary \ref{Choi thm 2}, for this example we have
\begin{equation}
C(i)=12h\,,\qquad l(i)=3\,,\qquad h_{i(0)}=0\,,\qquad h_{i(1)}=h_{i(2)}=h_{i(3)}=h\,. 
\end{equation}
More importantly, the bound for the ferromagnetic coupler strengths according to Corollary \ref{Choi thm 2} is given by
\begin{equation}
\label{Chois bound}
F_i<-8h\,.
\end{equation}
Our new tighter bound shows that a better bound exists. i.e.
$$
F_i<-6h\,,
$$
is sufficient for this toy model. See Appendix \ref{An example for the existence of a better bound} for details.

We will show later in Section \ref{Tightness of the bound} that the best bound for this example is $F_i<-5h$, if we allow $h_{i(k)}$ to have different values.

\subsection{Proof of the main theorem}
In this subsection, we give the full proof of our main theorem.
	
In order for sufficiently large $M_i$ to preserve the homogeneity of spins in $\iota(i)$, we need to find a sufficient condition so that the formation of each domain wall is forbidden. Now we have the following lemma.

\begin{lemma}
\label{lemma 1}
\begin{equation}
\label{lemma 1 main eqn}
M_i \ge \frac{1}{|\partial W_i|}\operatorname{min}\left\{|h(W_i)-J(W_i)|,|h(W_i)-h_i-J(\iota(i)\backslash W_i)| \right\}
\end{equation}
implies $\partial W_i$ is not a positive domain wall within the ground-state configuration of $\mathscr{E}_{I(G)}$.
\end{lemma}
\begin{proof}
Let $W_i(\pm)$ denote the spin configuration for all spins being $\pm 1$ in $W_i$ and $\overline{W}_i(\cdot)$ be the spin configuration for the complement of $W_i$ with respect to $\iota(i)$. Now suppose  $\partial W_i$ is a positive domain wall within the ground-state configuration of $\mathscr{E}_{I(G)}$. Then we have
\begin{equation}
\label{lemma 1 eqn 1}
\mathscr{E}_{I(G)}\left(W_i(+),\overline{W}_i(-),\dots\right)\le \mathscr{E}_{I(G)}\left(W_i(-),\overline{W}_i(-),\dots\right)\,,
\end{equation}
and
\begin{equation}
\label{lemma 1 eqn 2}
\mathscr{E}_{I(G)}\left(W_i(+),\overline{W}_i(-),\dots\right)\le \mathscr{E}_{I(G)}\left(W_i(+),\overline{W}_i(+),\dots\right)\,.
\end{equation}
However, according to equation \eqref{adiabatic Ising energy U}, we have
\begin{multline}
\mathscr{E}_{I(G)}\left(W_i(+),\overline{W}_i(-),\dots\right)
-
\mathscr{E}_{I(G)}\left(W_i(-),\overline{W}_i(-),\dots\right)
\\
=2\left(
\sum_{i(k)\in W_i} h_{i(k)}+\sum_{i(k)\in W_i}\sum_{l\in \operatorname{Onbh}(i(k))}J_{i(k)l}\, s_{\tau(l,i(k))}-\sum_{i_p i_q \in \partial W_i} F_i^{pq}
\right)
\\
\ge
2\left(
\sum_{i(k)\in W_i} h_{i(k)}-\sum_{i(k)\in W_i}\sum_{l\in \operatorname{Onbh}(i(k))} |J_{i(k)l}|-\sum_{i_p i_q \in \partial W_i} F_i^{pq}
\right)
\\
=
2\left(
h(W_i)-J(W_i)-\sum_{i_p i_q \in \partial W_i} F_i^{pq}
\right)
>
2\left(
h(W_i)-J(W_i)+|\partial W_i|\times M_i
\right)
\end{multline}
and
\begin{multline}
\mathscr{E}_{I(G)}\left(W_i(+),\overline{W}_i(-),\dots\right)
-
\mathscr{E}_{I(G)}\left(W_i(+),\overline{W}_i(+),\dots\right)
\\
=2\left(
-\sum_{j(k)\in \overline{W}_i} h_{j(k)}-\sum_{j(k)\in \overline{W}_i}\sum_{l\in \operatorname{Onbh}(j(k))}J_{j(k)l}\, s_{\tau(l,j(k))}-\sum_{i_p i_q \in \partial W_i} F_i^{pq}
\right)
\\
=2\left(
-\left(h_i-\sum_{i(k)\in W_i} h_{i(k)}\right)-\sum_{j(k)\in \overline{W}_i}\sum_{l\in \operatorname{Onbh}(j(k))}J_{j(k)l}\, s_{\tau(l,j(k))}-\sum_{i_p i_q \in \partial W_i} F_i^{pq}
\right)
\\
\ge
2\left(
h(W_i)-h_i-J(\overline{W}_i)-\sum_{i_p i_q \in \partial W_i} F_i^{pq}
\right)
>
2\left(
h(W_i)-h_i-J(\overline{W}_i)+|\partial W_i|\times M_i
\right)\,.
\end{multline}
Since our assumption also has
$$M_i \ge \frac{1}{|\partial W_i|}\operatorname{min}\left\{|h(W_i)-J(W_i)|,|h(W_i)-h_i-J(\iota(i)\backslash W_i)| \right\}\,,$$
we then have
\[
\mathscr{E}_{I(G)}\left(W_i(+),\overline{W}_i(-),\dots\right)
-
\mathscr{E}_{I(G)}\left(W_i(-),\overline{W}_i(-),\dots\right)
>0\,,
\]
or
\[
\mathscr{E}_{I(G)}\left(W_i(+),\overline{W}_i(-),\dots\right)
-
\mathscr{E}_{I(G)}\left(W_i(+),\overline{W}_i(+),\dots\right)
>0\,.
\]
This contradicts inequalities \eqref{lemma 1 eqn 1} and \eqref{lemma 1 eqn 2}. Hence $\left(W_i(+),\overline{W}_i(-),\dots\right)$ is not a positive domain wall within the ground-state configuration of $\mathscr{E}_{I(G)}$.
\end{proof}
Now we are ready to prove the main theorem.
\begin{proof}[Proof of the main theorem]
To prove
\begin{equation*}
s^*_{i_p}s^*_{i_q}=1\,, \qquad \text{for all $i_p i_q\in E(\iota(G))$}
\end{equation*}
and
$
\left(s^*_{1(1)},\dots,s^*_{1(|\iota(1)|)}\right)
$
is a ground-state configuration for Hamiltonian \eqref{adiabatic Ising U}, we can equivalently prove that no positive domain wall is present in the ground-state configuration. Note that the existence of a positive domain wall is equivalent to the existence of a domain wall.

Now, by Lemma \ref{lemma 1}, if $\emptyset\ne W_i \subsetneq \iota(i)$ and
\begin{equation*}
M_i \ge \frac{1}{|\partial W_i|}\operatorname{min}\left\{|h(W_i)-J(W_i)|,|h(W_i)-h_i-J(\iota(i)\backslash W_i)| \right\}
\end{equation*}	
we have that $W_i$ cannot have a positive domain wall $\partial W_i$ in the ground-state configuration. Therefore,
\begin{equation*}
M_i \ge \underset{\emptyset\ne W_i \subsetneq \iota(i)}{\operatorname{max}}\frac{1}{|\partial W_i|}\operatorname{min}\left\{|h(W_i)-J(W_i)|,|h(W_i)-h_i-J(\iota(i)\backslash W_i)| \right\}
\end{equation*}	
implies that no positive domain wall can be present in the ground-state configuration. Hence the ground-state configuration has no domain wall in $\iota(i)$.
\end{proof}

\section{Tightness of the bound}
\label{Tightness of the bound}
Now we want to show that, if the condition
\begin{equation}
\label{best constant condition}
h(W_i)\le h_i+J(\overline{W}_i)\qquad\text{or}\qquad h(\overline{W}_i)\le h_i-J(W_i)
\end{equation}
is satisfied, then
\begin{equation}
\label{best constant}
M(W_i;h,J):=\frac{1}{|\partial W_i|}\operatorname{min}\left\{|h(W_i)-J(W_i)|,|h(W_i)-h_i-J(\overline{W}_i)| \right\}
\end{equation}
is the best bound for $\emptyset\ne W_i\subsetneq \iota(i)$. That is for any $\epsilon>0$ and $F_i^{pq}=-M(W_i;h,J)+\epsilon$, we have that the ground-state of $\mathscr{E}_{I(G)}$ has a domain wall in $\iota(i)$ in the worst scenario. Here, the worst scenario is understood in the following theorem.
\begin{theorem}
\label{thm 2}
Suppose condition \eqref{best constant condition} is satisfied and let $M(W_i;h,J)$ be defined in equation \eqref{best constant}. For any $\epsilon>0$, if $F_i^{pq}=-M(W_i,h,J)+\frac{\epsilon}{|\partial W_i|}$, then
$\mathscr{E}_{I(G)}\left(W_i(+),\overline{W}_i(+),\dots\right)$ and $\mathscr{E}_{I(G)}\left(W_i(-),\overline{W}_i(-),\dots\right)$ are not the ground-state configurations for some values of $s_{\tau(i,j)}$ with $j\in \operatorname{nbh}(i)$.
\end{theorem}
Before giving the proof of Theorem \ref{thm 2}, we give some remarks and corollaries.
\begin{corollary}
	If
	\begin{equation}
	\label{best constant condition 1}
	h(\overline{W}_i)\le h_i+J(W_i)\qquad\text{or}\qquad h(W_i)\le h_i-J(\overline{W}_i)
	\end{equation}
	is satisfied, then $M(\overline{W}_i;h,J)$ is the tightest bound.
\end{corollary}
\begin{remark}
	\label{rmk of thm 2}
	\
	
	\begin{itemize}
	\item If condition \eqref{best constant condition} is satisfied for all non-empty $W_i\subsetneq \iota(i)$, then the right hand side of expression \eqref{main thm eqn 1} is the best constant. 
	\item If $h_i$ and $h(W_i)$ are both positive, then $M(W_i;h,J)$ is the best constant. Similarly, if $h_i$ and $h(\overline{W}_i)$ are both negative, then $M(\overline{W}_i;h,J)$ is the best constant. This can be checked easily via validity of condition \eqref{best constant condition} and \eqref{best constant condition 1} respectively.
	\end{itemize}
\end{remark}
Now we give an easy proof for the best constant for example \ref{An example}. \textit{The best bound for the example given in Subsection \ref{An example} is $5h$.} Recall in Remark \ref{main thm rmk} that we need to relax the assignment of $h$-fields. Moreover, in this example, we have only one non-trivial embedding (the green vertices in Figure \ref{minorexmaple2}) and $h_i=3h>0$. By Remark \ref{rmk of thm 2}, the best bound is given by $M_i=5h$, if we allow a more general distribution of $h_{i(k)}$. See Appendix \ref{Best bound on the example} for details.

Now we give the proof of Theorem \ref{thm 2}.
\begin{proof}[Proof of Theorem \ref{thm 2}]
As in the proof of the previous lemma, one has
\begin{multline*}
\mathscr{E}_{I(G)}\left(W_i(+),\overline{W}_i(-),\dots\right)
-
\mathscr{E}_{I(G)}\left(W_i(-),\overline{W}_i(-),\dots\right)
\\
=2\left(
\sum_{i(k)\in W_i} h_{i(k)}+\sum_{i(k)\in W_i}\sum_{l\in \operatorname{Onbh}(i(k))}J_{i(k)l}\, s_{\tau(l,i(k))}-\sum_{i_p i_q \in \partial W_i} F_i^{pq}
\right)
\,.
\end{multline*}
For some $s_{\tau(l,i(k))}$ with $i(k)\in V(W_i)$, we have
\begin{multline}
\label{lemma 2 eqn 1}
\mathscr{E}_{I(G)}\left(W_i(+),\overline{W}_i(-),\dots\right)
-
\mathscr{E}_{I(G)}\left(W_i(-),\overline{W}_i(-),\dots\right)
\\
=
2\left(
h(W_i)-J(W_i)-\sum_{i_p i_q \in \partial W_i} F_i^{pq}
\right)
=
2\left(
h(W_i)-J(W_i)+|\partial W_i|\times M(W_i,h,J)-\epsilon
\right)
\\
\le
2\left(
h(W_i)-J(W_i)+|h(W_i)-J(W_i)|-\epsilon
\right)
\,.
\end{multline}
\noindent\textbf{Case 1:} $h(W_i)-J(W_i)\ge 0\,.$
Let us consider the following difference
\begin{multline*}
\mathscr{E}_{I(G)}\left(\overline{W}_i(+),W_i(-),\dots\right)
-
\mathscr{E}_{I(G)}\left(W_i(-),\overline{W}_i(-),\dots\right)
\\
=2\left(
\sum_{i(k)\in \overline{W}_i} h_{i(k)}+\sum_{j(k)\in \overline{W}_i}\sum_{l\in \operatorname{Onbh}(j(k))}J_{j(k)l}\, s_{\tau(l,j(k))}-\sum_{i_p i_q \in \partial W_i} F_i^{pq}
\right)
\,.
\end{multline*}
For some $s_{\tau(l,j(k))}$ with $j(k)\in V(\overline{W}_i)$, we have
\begin{multline}
\mathscr{E}_{I(G)}\left(\overline{W}_i(+),W_i(-),\dots\right)
-
\mathscr{E}_{I(G)}\left(W_i(-),\overline{W}_i(-),\dots\right)
\\
=
2\left(
h(\overline{W}_i)-J(\overline{W}_i)-\sum_{i_p i_q \in \partial W_i} F_i^{pq}
\right)
=
2\left(
h(\overline{W}_i)-J(\overline{W}_i)+|\partial W_i|\times M(W_i,h,J)-\epsilon
\right)
\\
\le
2\left(
h(\overline{W}_i)-J(\overline{W}_i)+|h(W_i)-J(W_i)|-\epsilon
\right)
\\
=
2\left(
h_i-J(\overline{W}_i)-J(W_i)-\epsilon
\right)
\le
-2\epsilon
<0\,.
\end{multline}
Note that we used the fact that $h_i\le |h_i| \le J(\overline{W}_i)+J(W_i)$ in the last step. Therefore, $\left(W_i(-),\overline{W}_i(-),\dots\right)$ is not a ground-state configuration. Moreover, one can show that
\begin{multline}
\mathscr{E}_{I(G)}\left(W_i(-),\overline{W}_i(+),\dots\right)
-
\mathscr{E}_{I(G)}\left(W_i(+),\overline{W}_i(+),\dots\right)
\\
=
2\left(
-h(W_i)-\sum_{i(k)\in W_i}\sum_{l\in \operatorname{Onbh}(i(k))}J_{i(k)l}\, s_{\tau(l,i(k))}
-\sum_{i_p i_q \in \partial W_i} F_i^{pq}
\right)
\\
=
2\left(
-h(W_i)-\sum_{i(k)\in W_i}\sum_{l\in \operatorname{Onbh}(i(k))}J_{i(k)l}\, s_{\tau(l,i(k))}+|\partial W_i|\times M(W_i,h,J)-\epsilon
\right)
\\
\le
2\left(
-h(W_i)+J(W_i)+|h(W_i)-J(W_i)|-\epsilon
\right)
=-2\epsilon
<0
\end{multline}
Hence $\left(W_i(+),\overline{W}_i(+),\dots\right)$ is also not a ground-state configuration.

\

\noindent\textbf{Case 2:} $h(W_i)-J(W_i)< 0\,.$
We can easily see from equation \eqref{lemma 2 eqn 1} that $\left(W_i(-),\overline{W}_i(-),\dots\right)$ is not a ground-state configuration.

\noindent Now we show that $\left(W_i(+),\overline{W}_i(+),\dots\right)$ is also not a ground-state configuration. The proof is similar to the previous case, but one needs to take care of the extra asymmetry caused by $h_i$. Let us start with the following expression
\begin{multline*}
\mathscr{E}_{I(G)}\left(W_i(+),\overline{W}_i(-),\dots\right)
-
\mathscr{E}_{I(G)}\left(W_i(+),\overline{W}_i(+),\dots\right)
\\
=2\left(
-\sum_{j(k)\in \overline{W}_i} h_{j(k)}-\sum_{j(k)\in \overline{W}_i}\sum_{l\in \operatorname{Onbh}(j(k))}J_{j(k)l}\, s_{\tau(l,j(k))}-\sum_{i_p i_q \in \partial W_i} F_i^{pq}
\right)
\,.
\end{multline*}
For some $s_{\tau(l,j(k))}$ with $j(k)\in V(\overline{W}_i)$, we have
\begin{multline}
\label{lemma 2 eqn 2}
\mathscr{E}_{I(G)}\left(W_i(+),\overline{W}_i(-),\dots\right)
-
\mathscr{E}_{I(G)}\left(W_i(+),\overline{W}_i(+),\dots\right)
\\
=
2\left(
-h(\overline{W}_i)-J(\overline{W}_i)-\sum_{i_p i_q \in \partial W_i} F_i^{pq}
\right)
=
2\left(
h(W_i)-h_i-J(\overline{W}_i)+|\partial W_i|\times M(W_i,h,J)-\epsilon
\right)
\\
\le
2\left(
h(W_i)-h_i-J(\overline{W}_i)+|h(W_i)-h_i-J(\overline{W}_i)|-\epsilon
\right)
\,.
\end{multline}
\textbf{Case 2.1:} If $h(W_i)-h_i-J(\overline{W}_i)\le 0\,,$
we can see from equation \eqref{lemma 2 eqn 2} that $\left(W_i(+),\overline{W}_i(+),\dots\right)$ is not a ground-state configuration.

\noindent\textbf{Case 2.2:} If $h(W_i)-h_i-J(\overline{W}_i)>0\,,$ then by condition \eqref{best constant condition}, one must have
\[
h(\overline{W}_i)\le h_i-J(W_i)\,,
\]
which is equivalent to
\[
h(W_i)-J(W_i)\ge 0\,.
\]
Therefore, following the same as Case 1, we complete the proof.


\end{proof}

\subsection{Admissible minor embeddings}
\label{Admissible minor embeddings}
Now we show that conditions \eqref{best constant condition} and \eqref{best constant condition 1} should be satisfied for any reasonable minor embedding. We call a minor embedding, say $(I,h,J,F)$, admissible if the following condition is satisfied.
\begin{itemize}
	\item $(I,h,J,F)$ does not exclude any possible spin configuration for any $i\in G$ in any embedded Ising problem.
\end{itemize}
Here $F$ denotes the absolute value of the chain strength. Note that admissible minor embeddings are more suitable for practical purposes, since for general NP-hard problems we do not expect any pre-assignment for any logical qubit in $G$. It can be shown that the condition for admissible minor embeddings implies conditions \eqref{best constant condition} and \eqref{best constant condition 1}.
\begin{proof}[Verification]
\begin{equation*}
\neg\text{condition \eqref{best constant condition}} \vee \neg\text{condition \eqref{best constant condition 1}}
\end{equation*}
is equivalent to
\begin{multline*}
\left[\left(-J(\overline{W}_i)>h(\overline{W}_i)\right) \wedge \left(J(W_i)>h(W_i)\right)\right] \vee \left[\left(-J(W_i)>h(W_i)\right) \wedge \left(J(\overline{W}_i)>h(\overline{W}_i)\right)\right]\\
\implies h(W_i)<J(W_i) \quad\text{for some $W_i\subset \iota(i)$.}
\end{multline*}
By the \textbf{Case 2} analysis in the proof of theorem \ref{thm 2}, we see that $\left(W_i(+),\overline{W}_i(+),\dots\right)$ is the only possible ground-state configuration for some problems, if $F_i^{pq}>-M(W_i,h,J)$. This is a pre-assignment for the $i$-th logical qubit. Hence it is not an admissible minor embedding.
\end{proof}
Now, an immediate consequence of Theorem \ref{thm 2} gives
\begin{theorem}
	\label{corollary of thm 2}
	Let $(I,h,J,F)$ be an admissible minor embedding and $M(W_i;h,J)$ be defined in equation \eqref{best constant}. Then $M(W_i;h,J)$ is the best constant for all $W_i\subset \iota(i)$. Hence $$\underset{W_i}{\operatorname{max}}\left( \frac{1}{|\partial W_i|}\operatorname{min}\Big\{|h(W_i)-J(W_i)|,|h(W_i)-h_i-J(\iota(i)\backslash W_i)| \Big\}\right)$$ is the tightest bound for admissible minor embeddings.
\end{theorem}
\begin{remark}[Importance of the distribution of $h_{i(k)}$]\label{rmk on broken chains}
	An admissible minor embedding $(I,h,J,F)$ can be viewed as a minimum requirement for perfect (non-broken) chains in the worst scenario. The minimum strength of $F_i^{pq}$ is determined by $h_{i(k)}$ and $J$ via the expression of $M(W_i;h,J)$. However, if we fix the values of the $F_i^{pq}$'s, we cannot choose the distribution of $h_{i(k)}$ arbitrarily, even with condition \eqref{embedding condition 1.1} ($\sum h_{i(k)}= h_i$) satisfied. This will not cause any trouble if the $F_i^{pq}$'s are sufficiently large. However, when the $F_i^{pq}$'s are small compared with $h_{i(k)}$, one needs to be more careful. More precisely, if we define $C(W_i):=J(W_i)+|\partial (W_i)|\times F-|h(W_i)|$, then $C(W_i)$ has to be greater or equal to zero for admissible minor embeddings. In other words, we must have $|h(W_i)|\le J(W_i)+|\partial (W_i)|\times F$, which is an upper bound for $h_{i(k)}$. This condition can be easily violated when $h_{i(k)}$ is concentrated in a single physical qubit and $F$ is comparably small. This is the situation when we apply the single distribution method as defined in \cite{Pudenz}. Therefore, there are likely to be some non-admissible minor embeddings in the single distribution method.

\end{remark}
%
%
%

%
\section{Experimental results}
\label{Experimental results}
\begin{figure}[h]
	\centering
	\includegraphics[scale=1.1]{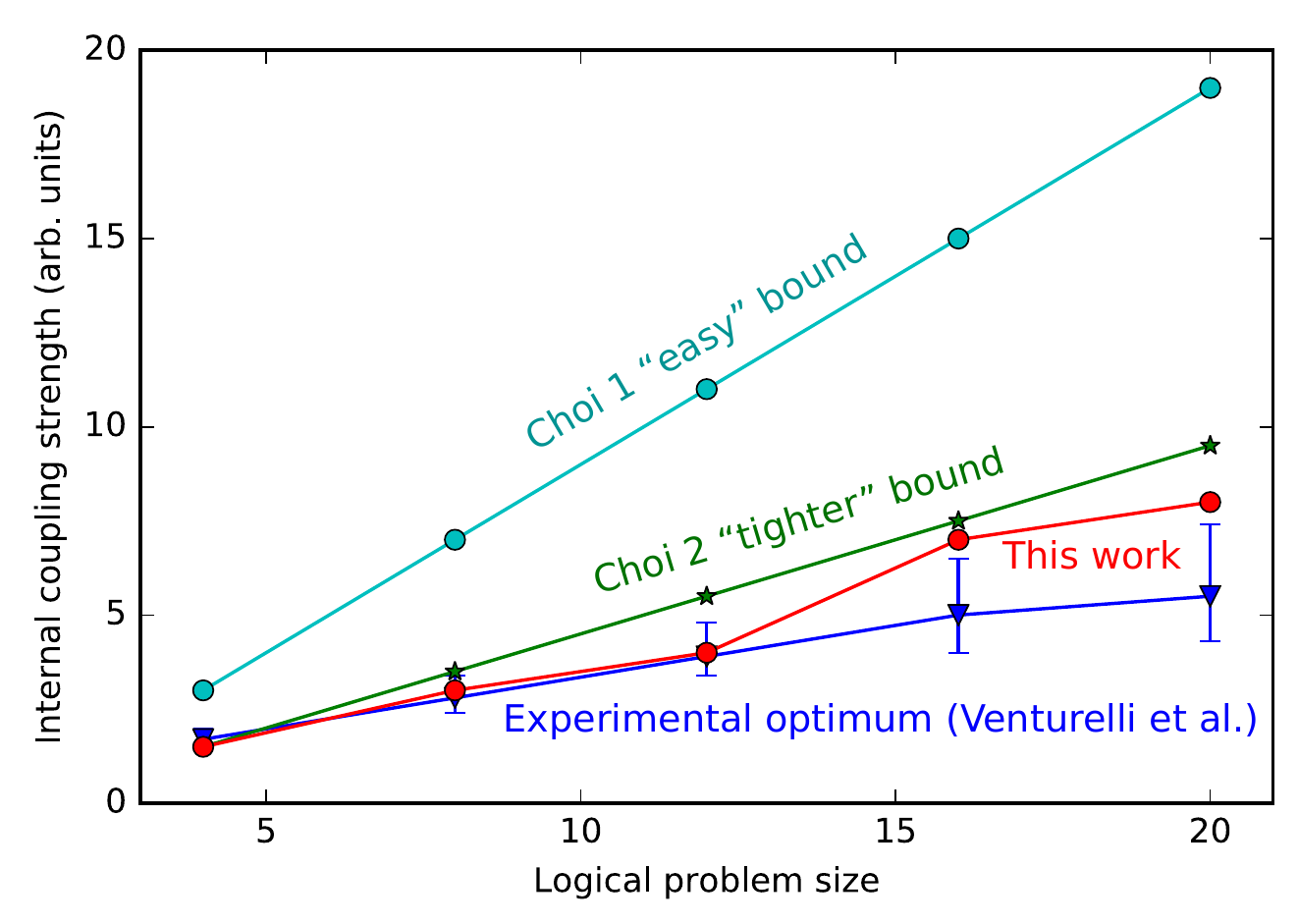}
	\caption{Experimentally-determined optimum magnitude of the internal coupling strength for the Sherrington-Kirkpatrick model on a complete graph (dark blue extracted from \cite{Venturelli2015}). The red points are bounds on the magnitude of the internal coupling strength obtained using the method introduced in Section \ref{Main results}. The pale blue and green points are obtained by Choi's first and second method respectively in \cite{Choi2008}. The lines are a guide to the eye}
	\label{Different bounds for Sherrington-Kirkpatrick spin glass on the D--Wave}
\end{figure}
In this subsection, we will compare different methods for estimating the optimum internal coupling strength to show how close they are to the experimental optima. We use the experimental data from Venturelli et.al \cite{Venturelli2015}, where fully connected Sherrington-Kirkpatrick spin-glass problems are implemented on the D--Wave DW2X machine. As we are only interested in optimal values of the internal coupling strength without broken chains, we extract the optimal values without any majority-vote post-processing. As we can see from Figure \ref{Different bounds for Sherrington-Kirkpatrick spin glass on the D--Wave}, our new tighter bound approaches more closely to the true experimental optima.

\section{Conclusions and future work}
\label{Conclusions and future work}
There are many challenges for realising a quantum annealer capable of outperforming classical computation for some classes of problems. Our work shows the importance of optimal ferromagnetic coupling strength and gives the best theoretical bound in our main theorem \ref{main thm}. However, this is valid under the condition given in our second theorem \ref{thm 2}. In fact, we can give the best bound when the logical qubit has non-negative $h_{i(k)}$-fields. Our bound is certainly tighter than Choi's bounds as shown in our toy example \ref{An example}. We have introduced the concept of admissible minor embeddings, which means that condition \eqref{embedding condition 1.1} ($\sum h_{i(k)} = h_i$) is not sufficient to guarantee an admissible minor embedding when $F_i^{pq}$ is small compared with $h_{i(k)}$. Note that having an admissible minor embedding is necessary for practical reasons. For non-admissible minor embeddings, one could in theory achieve a better bound and obtain a correct ground-state under quantum annealing, but this requires a pre-knowledge of the ground-state configuration of logical problem. 

Experimental results from quantum annealers show that our new method can be used to reduce the time-to-solution. However, this comes at a cost. The computational effort to calculate our new bound is $O(D2^L)$ per logical qubit, where $D$ is the degree of the logical qubit and $L$ is the chain length. Note that for Choi's two bounds, the computational effort are $O(D)$ and $O(DL)$ respectively.
Finally, it still remains open how to assign admissible $h_{i(k)}$-fields to yield the best performance on actual quantum annealers.

\section*{Acknowledgements}
The research is based upon work (partially) supported by EPSRC (grant reference EP/R020159/1) and the Office of the Director of National Intelligence (ODNI), Intelligence Advanced Research Projects Activity (IARPA), via the U.S. Army Research Office contract W911NF-17-C-0050. The views and conclusions contained herein are those of the authors and should not be interpreted as necessarily representing the official policies or endorsements, either expressed or implied, of the ODNI, IARPA, or the U.S. Government. The U.S. Government is authorized to reproduce and distribute reprints for Governmental purposes notwithstanding any copyright annotation thereon.

\begin{appendices}
\section{Job-shop scheduling problems on the D--Wave 2000Q Machine}
\label{Job-shop scheduling problems on the D-Wave Machine}
We will now show some experimental results obtained on the D-Wave quantum annealer. These illustrate the dependence on the internal coupling strength $F_i^{pq}$ and show that there is an optimum value for it. In this subsection, we will use the performance of the NP-hard job-shop scheduling problem (JSP) on the D--Wave 2000Q to illustrate the importance of the best bound. Here we will follow the methodology introduced by the NASA Ames team \cite{Rieffel2014,Rieffel2015,Venturelli2016}. We will use time-to-solution as a benchmarking metric.

A typical job-shop scheduling problem (JSP) consists of a set of N jobs $J = \{j_1,\dots, j_N \}$ that must be
scheduled on a set of machines $M =\{ m_1,\dots, m_P \}$. Each job consists of a sequence of operations that
must be performed in a predefined order $j_n = \{ O_{n,1}\to O_{n,2} \to \dots \to O_{n,L_n} \}$, where each job jn has $L_n$
operations. Each operation $O_{n,k}$ has a non-negative integer execution time $\tau_{n,k}$ and has to be executed
by an assigned machine $m_{n,k}\in M$. The goal of solving JSP is to find an optimal scheduling that
minimises the makespan, i.e. the minimum time to finish all the jobs.

A generalised tabular representation of job shop scheduling problems is shown in Table 1.
\begin{table}[h]
	\centering
	\caption{M-table and P-table for JSP}
	\begin{subtable}{.5\linewidth}
		\centering
		\caption{Machine allocation}
		\begin{tabular}{| c | c | c | c | c | }
			\hline
			& $\text{Operation}_{*,1}$ & $\text{Operation}_{*,2}$ & $\dots$ & $\text{Operation}_{*,K}$ \\ \hline
			$\text{j}_1$ & $\text{m}_{1,1}$ & $\text{m}_{1,2}$ & $\dots$ & $\text{m}_{1,K}$\\ \hline
			$\text{j}_2$ & $\text{m}_{2,1}$ & $\text{m}_{2,2}$ & $\dots$ & $\text{m}_{2,K}$\\ \hline
			$\vdots$ & $\vdots$ & $\vdots$ & $\ddots$ & $\vdots$ \\ \hline
			$\text{j}_N$ & $\text{m}_{N,1}$ & $\text{m}_{N,2}$ & $\dots$ & $\text{m}_{N,K}$\\
			\hline
		\end{tabular}
	\end{subtable}%
	\\
	\begin{subtable}{.5\linewidth}
		\centering
		\caption{Time (per unit) spent on each operation}
		\begin{tabular}{| c | c | c | c | c | }
			\hline
			& $\text{Operation}_{*,1}$ & $\text{Operation}_{*,2}$ & $\dots$ & $\text{Operation}_{*,K}$ \\ \hline
			$\text{j}_1$ & $\tau_{1,1}$ & $\tau_{1,2}$ & $\dots$ & $\tau_{1,K}$\\ \hline
			$\text{j}_2$ & $\tau_{2,1}$ & $\tau_{2,2}$ & $\dots$ & $\tau_{2,K}$\\ \hline
			$\vdots$ & $\vdots$ & $\vdots$ & $\ddots$ & $\vdots$ \\ \hline
			$\text{j}_N$ & $\tau_{N,1}$ & $\tau_{N,2}$ & $\dots$ & $\tau_{N,K}$\\
			\hline
		\end{tabular}
	\end{subtable} 
\end{table}

\noindent For any job-shop scheduling problem, we can easily write it in the above representation by setting $\tau_{n,k}=0$ for non-given operations and $K=\max\limits_n L_n$. 
To translate the problem into an Ising Hamiltonian, we follow the method proposed by Venturelli et al. \cite{Venturelli2016} and assign a set of binary variables for each operation, corresponding to the various possible discrete starting times the operation can have:
$$
x_{n,k;t}=\left\{
\begin{aligned}
&1\,:\qquad \text{operation}\,\, O_{n,k}\,\, \text{starts at time}\,\, t\,,\\
&0\,:\qquad \text{otherwise}\,.
\end{aligned}
\right.
$$
Here $t$ is bounded from above by the timespan $T$, which represents the maximum time we allow for all jobs to be completed.
The resulting classical objective function (Hamiltonian) is given by
\begin{equation}
\label{total H}
H_T(x)=E_{\text{problem}} \left(h_1(x)+h_2(x) +h_3(x) + h_4(x)\right)\,, 
\end{equation}
where $E_{\text{problem}}$ is the energy scaling parameter and each penalty term is explained briefly as follows.

\begin{itemize}
	\item $h_1(x)=\sum_{n,k}\left(\sum_tx_{n,k;t}-1 \right)^2\,$, checks that an operation must start once and only once.\\
	\item $h_2(x)=\sum_{n}\sum_{k<n}\left(\sum_{t+\tau_{n,k}>t^\prime}\, x_{n,k;t}\,x_{n,k+1;t^\prime}\right)\,$, ensures that the order of the operations within a job is preserved.\\
	\item $h_3(x)=\sum_{t+\tau_{nK}>T}\,x_{n,K;t}\,$, guarantees that the last operation in each job finishes by time $T$.\\
	\item $h_4(x)=\sum_{m}\left(\sum_{(n,k;t|n^\prime,k^\prime;t^\prime)\in R_m}x_{n,k;t}\,x_{n^\prime,k^\prime;t^\prime}\right)\,$, $R_m$ consists of two penalty sets given in the following.
	\begin{itemize}
		\item Forbidding operation $O_{n^\prime,k^\prime}$ from starting at $t^\prime$ if there is another operation $O_{n,k}$ still running.
		\item Two operations cannot start at the same time, unless at least one of them has an execution time equal to zero\,.
	\end{itemize}
\end{itemize}
Due to the detailed structure of the JSP Hamiltonian, we have (from equation \eqref{def C}):
\begin{equation*}
C(i)=\frac{1}{2}E_{problem}\,,
\end{equation*}
and the spectral gap is given by
\begin{equation*}
\Delta=E_{problem}\,.
\end{equation*}
Hence, an easy follow-up from Corollary \ref{Choi thm 2} can be derived (or see \cite{Choi2008}). i.e. If topological embeddings are chosen to embed the job shop scheduling problem Hamiltonian, we find that $|F|\ge \frac{1}{2}(C(i)+\Delta)=\frac{3}{4} E_{problem}$ is a sufficient lower bound which preserves the spectral gap of the original Hamiltonian.

\begin{figure}[h]
	\centering
	\includegraphics[width=.9\textwidth]{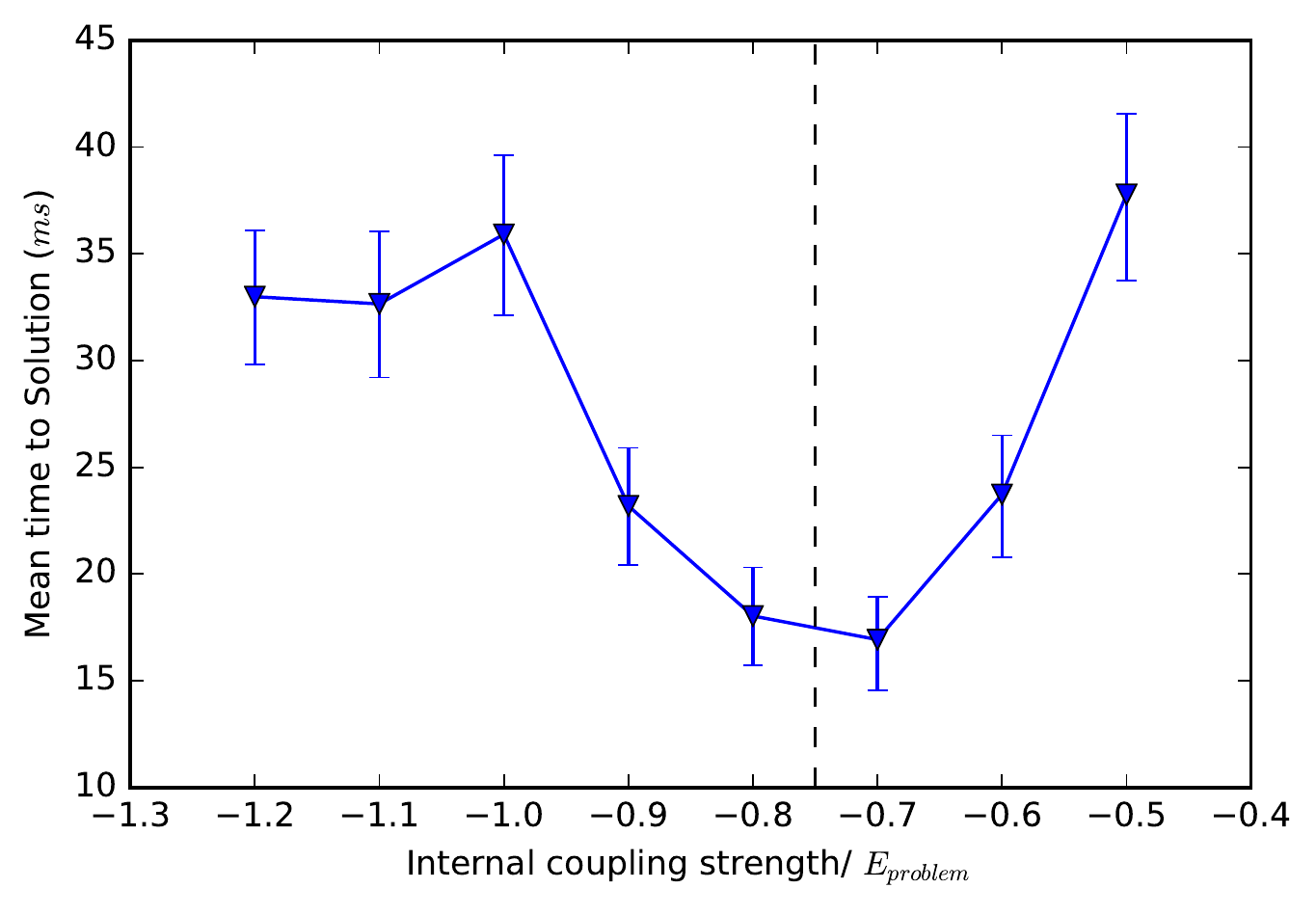}
	\caption{The graph shows the dependence of the time-to-solution on the internal coupling strength for solving a job shop scheduling problem on the D-Wave 2000Q. See the text for details. The dashed line shows the calculated value of the optimum internal coupling strength using the method of \cite{Choi2008}. The error bars are obtained by bootstrapping with $95\%$ confidence intervals.}
	\label{Dwaveresult}
\end{figure}
Theorem \ref{main thm} and Corollary \ref{Choi thm 2} is based on an ideal quantum annealer. It is clear that $l(i)$ depends only
on the number of leaves in sub-trees of a minor embedding, which is independent of the lengths of
branches within the trees. This means that in the ideal case there is no difference between short chains and long chains as long
as equations \eqref{Choi thm 2 eqn 1} and \eqref{Choi thm 2 eqn 2} are satisfied. However, due to engineering limitations, there is an upper
bound, say $\lambda$, for both logical and internal coupling strengths in the actual machine. Therefore, one has to rescale
(i.e. decrease) the strength of the logical interaction in order for it to fit into the confined range. This leads us to the existence of an optimal coupling strength for chains in reality. 

Figure \ref{Dwaveresult} shows the importance of the optimal bound in the D-Wave 2000Q machine, as the shortest time to solution is achieved close to the theoretical bound that we derived in the previous sections. The data is obtained by running 200 random JSPs with size $N=3$, $K=3$ and $T=8$ on the D-Wave 2000Q machine. For each instance five minor embeddings are randomly generated. At each value of the internal coupling strength the probability of finding the correct JSP solution is experimentally determined by running the annealer 10,000 times for each embedding. The time-to-solution (TTS) is defined as the expected time taken to find the solution with probability $p=99.9\%$ and is given by \cite{Ronnow2014}:
$$
TTS=t_a \left( \frac{\log[1-p]}{\log [1-s]}  \right)\,,
$$
where $s$ is the success probability for each embedding and $t_a$ is the single-run annealing time, which is equal to $2\mu s$ in our experiments. For each instance the minimum TTS for the five embeddings is recorded. The same procedure is conducted for the 200 random instances and then the mean TTS is the data shown in Figure \ref{Dwaveresult}. Error bars are obtained by bootstrapping method and the confidence intervals are chosen to be $95\%$.


We expect that the theoretical optimal bound plays an important role in a general quantum annealer and it is not constrained to JSPs.

\section{An example for the existence of a better bound}
\label{An example for the existence of a better bound}
Here we show that tighter bounds exists then those given in \cite{Choi2008} by continuing the toy example of Figure \ref{minorexmaple1}. According to Corollary \ref{Choi thm 2 eqn 2}, the assignments of local $h_{i(k)}$ are given as in Figure \ref{minorexmaple2}.
\begin{figure}[h]
	\centering
	\includegraphics[scale=0.3]{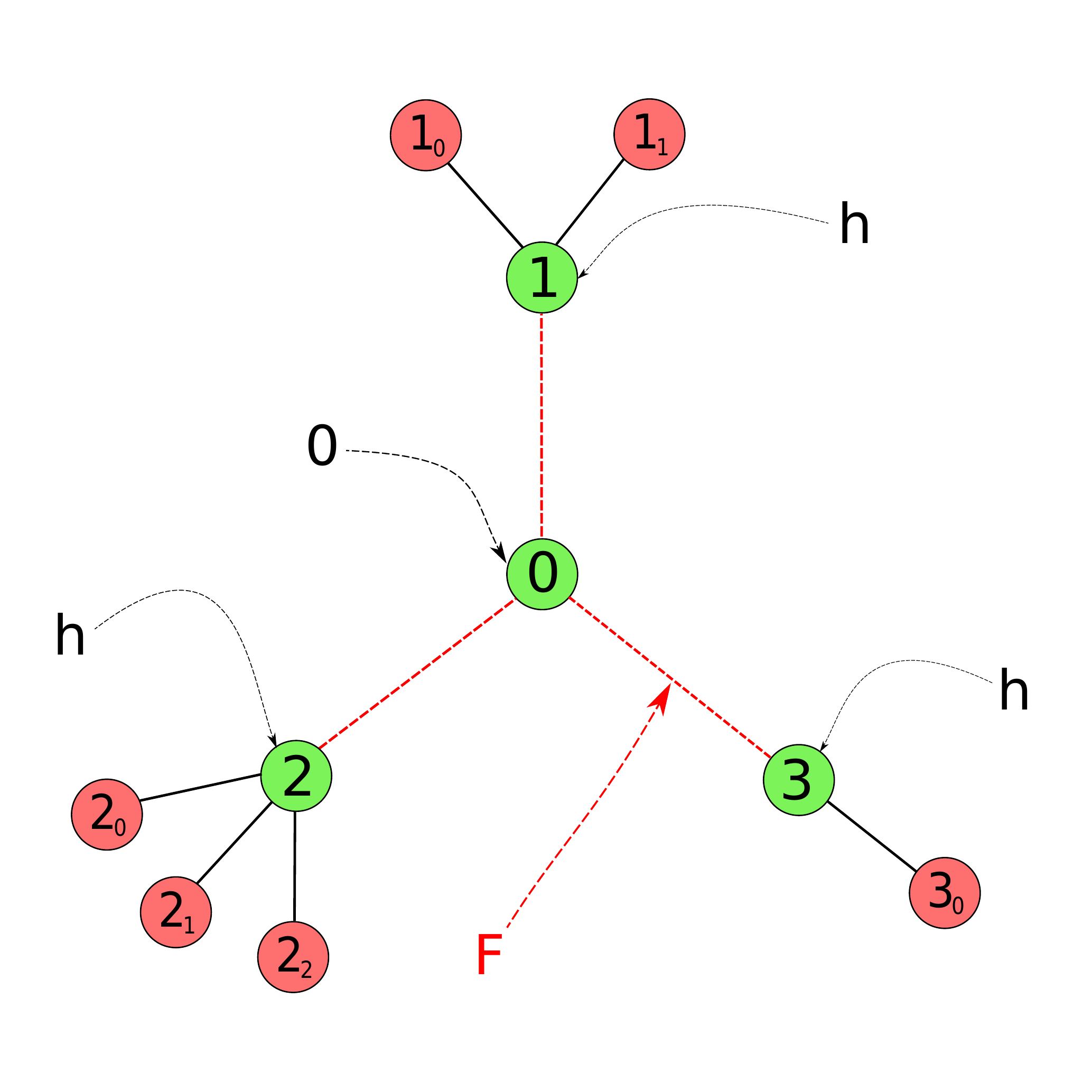}
	\caption{An example of $\iota(i)$}
	\label{minorexmaple2}
\end{figure}
Let $\left(\begin{smallmatrix} &s_1&\\&s_0& \\ s_2& &s_3\end{smallmatrix}\right)$ denote the assignments of spin values for vertices $0, 1, 2$ and $3$. For example
\[
\left(\begin{smallmatrix} &s_1&\\&s_0& \\ s_2& &s_3\end{smallmatrix}\right)
=\left(\begin{smallmatrix} &-&\\&+& \\ +& &+\end{smallmatrix}\right)
\]
means that the spin value is $-1$ for vertex $1$ and the spin values are equal to $+1$ for the other vertices.

\subsubsection*{Case 1 inequality}
Now we have the following inequalities.
\begin{equation}
\label{case1-1}
\frac{1}{2}\left[
\mathscr{E}\left(\begin{smallmatrix} &-&\\&+& \\ +& &+\end{smallmatrix}\right)
-
\mathscr{E}\left(\begin{smallmatrix} &-&\\&-& \\ -& &-\end{smallmatrix}\right)
\right]
\ge
2\times h-2\times 5h - F
=
-8h - F
\end{equation}
and
\begin{equation}
\label{case1-2}
\frac{1}{2}\left[
\mathscr{E}\left(\begin{smallmatrix} &-&\\&+& \\ +& &+\end{smallmatrix}\right)
-
\mathscr{E}\left(\begin{smallmatrix} &+&\\&+& \\ +& &+\end{smallmatrix}\right)
\right]
\ge
-h -5h -F
=
-6h - F
\,.
\end{equation}
If the configuration $\left(\begin{smallmatrix} &-&\\&+& \\ +& &+\end{smallmatrix}\right)$ is not part of the ground-state configuration, then we must have the right hand side of either inequality \eqref{case1-1} or inequality \eqref{case1-2} greater than zero. That is
\begin{equation}
\label{case1 inequality}
F<-6h\,.
\end{equation}
Due to the symmetric property of our example, we have that $\left(\begin{smallmatrix} &+&\\&+& \\ -& &+\end{smallmatrix}\right)$ and $\left(\begin{smallmatrix} &+&\\&+& \\ +& &-\end{smallmatrix}\right)$ cannot be part of the ground-state configuration if $F<-6h$.

\subsubsection*{Case 2 inequality}
Using the same method, one can derive that
\begin{equation}
\label{case2-1}
\frac{1}{2}\left[
\mathscr{E}\left(\begin{smallmatrix} &+&\\&-& \\ -& &-\end{smallmatrix}\right)
-
\mathscr{E}\left(\begin{smallmatrix} &-&\\&-& \\ -& &-\end{smallmatrix}\right)
\right]
\ge
h- 5h - F
=
-4h - F
\end{equation}
and
\begin{equation}
\label{case2-2}
\frac{1}{2}\left[
\mathscr{E}\left(\begin{smallmatrix} &+&\\&-& \\ -& &-\end{smallmatrix}\right)
-
\mathscr{E}\left(\begin{smallmatrix} &+&\\&+& \\ +& &+\end{smallmatrix}\right)
\right]
\ge
-2\times h -2\times 5h -F
=
-12h - F
\,.
\end{equation}
That is 
\begin{equation}
\label{case2 inequality}
F<-4h\,.
\end{equation}
Again, due to the symmetric property of our example, we have that $\left(\begin{smallmatrix} &-&\\&-& \\ +& &-\end{smallmatrix}\right)$ and $\left(\begin{smallmatrix} &-&\\&-& \\ -& &+\end{smallmatrix}\right)$ cannot be part of the ground-state configuration if $F<-4h$.

\subsubsection*{Case 3 inequality}
Using the same method, one can derive that
\begin{equation}
\label{case3-1}
\frac{1}{2}\left[
\mathscr{E}\left(\begin{smallmatrix} &+&\\&-& \\ +& &-\end{smallmatrix}\right)
-
\mathscr{E}\left(\begin{smallmatrix} &-&\\&-& \\ -& &-\end{smallmatrix}\right)
\right]
\ge
2\times h -2\times 5h -2\times F
=
-12h - 2F
\end{equation}
and
\begin{equation}
\label{case3-2}
\frac{1}{2}\left[
\mathscr{E}\left(\begin{smallmatrix} &+&\\&-& \\ +& &-\end{smallmatrix}\right)
-
\mathscr{E}\left(\begin{smallmatrix} &+&\\&+& \\ +& &+\end{smallmatrix}\right)
\right]
\ge
-h - 5h -2\times F
=
-6h - 2F
\,.
\end{equation}
That is 
\begin{equation}
\label{case3 inequality}
F<-3h\,.
\end{equation}
The symmetric property of our example tells us that $\left(\begin{smallmatrix} &-&\\&-& \\ +& &+\end{smallmatrix}\right)$ and $\left(\begin{smallmatrix} &+&\\&-& \\ -& &+\end{smallmatrix}\right)$ cannot be part of the ground-state configuration if $F<-3h$.

\subsubsection*{Case 4 inequality}
Using the same method, one can derive that
\begin{equation}
\label{case4-1}
\frac{1}{2}\left[
\mathscr{E}\left(\begin{smallmatrix} &+&\\&+& \\ -& &-\end{smallmatrix}\right)
-
\mathscr{E}\left(\begin{smallmatrix} &-&\\&-& \\ -& &-\end{smallmatrix}\right)
\right]
\ge
h - 5h -2\times F
=
-4h - 2F
\end{equation}
and
\begin{equation}
\label{case4-2}
\frac{1}{2}\left[
\mathscr{E}\left(\begin{smallmatrix} &+&\\&+& \\ -& &-\end{smallmatrix}\right)
-
\mathscr{E}\left(\begin{smallmatrix} &+&\\&+& \\ +& &+\end{smallmatrix}\right)
\right]
\ge
-2\times h - 2\times 5h -2\times F
=
-12h - 2F
\,.
\end{equation}
That is 
\begin{equation}
\label{case4 inequality}
F<-2h\,.
\end{equation}
According to the symmetric property of our example, we have that $\left(\begin{smallmatrix} &-&\\&+& \\ +& &-\end{smallmatrix}\right)$ and $\left(\begin{smallmatrix} &-&\\&+& \\ -& &+\end{smallmatrix}\right)$ cannot be part of the ground-state configuration if $F<-2h$.

\subsubsection*{Case 5 inequality}
Using the same method, one can derive that
\begin{equation}
\label{case5-1}
\frac{1}{2}\left[
\mathscr{E}\left(\begin{smallmatrix} &-&\\&+& \\ -& &-\end{smallmatrix}\right)
-
\mathscr{E}\left(\begin{smallmatrix} &-&\\&-& \\ -& &-\end{smallmatrix}\right)
\right]
\ge
0\times h - 3\times F
=
- 3F
\end{equation}
and
\begin{equation}
\label{case5-2}
\frac{1}{2}\left[
\mathscr{E}\left(\begin{smallmatrix} &-&\\&+& \\ -& &-\end{smallmatrix}\right)
-
\mathscr{E}\left(\begin{smallmatrix} &+&\\&+& \\ +& &+\end{smallmatrix}\right)
\right]
\ge
-3\times h - 3\times 5h -3\times F
=
-18h - 3F
\,.
\end{equation}
That is 
\begin{equation}
\label{case5 inequality}
F<0\,.
\end{equation}

\subsubsection*{Case 6 inequality}
Using the same method, one can derive that
\begin{equation}
\label{case6-1}
\frac{1}{2}\left[
\mathscr{E}\left(\begin{smallmatrix} &+&\\&-& \\ +& &+\end{smallmatrix}\right)
-
\mathscr{E}\left(\begin{smallmatrix} &-&\\&-& \\ -& &-\end{smallmatrix}\right)
\right]
\ge
3\times h -3\times 5h - 3\times F
=
-12h - 3F
\end{equation}
and
\begin{equation}
\label{case6-2}
\frac{1}{2}\left[
\mathscr{E}\left(\begin{smallmatrix} &+&\\&-& \\ +& &+\end{smallmatrix}\right)
-
\mathscr{E}\left(\begin{smallmatrix} &+&\\&+& \\ +& &+\end{smallmatrix}\right)
\right]
\ge
-0\times h -3\times F
=
- 3F
\,.
\end{equation}
That is 
\begin{equation}
\label{case6 inequality}
F<0\,.
\end{equation}
Now from inequalities \eqref{case1 inequality}, \eqref{case2 inequality}, \eqref{case3 inequality}, \eqref{case4 inequality}, \eqref{case5 inequality} and \eqref{case6 inequality}, we have that if
\begin{equation}
\label{allcases inequality}
F<-6h\,,
\end{equation}
only homogeneous configurations within $\iota(i)$ (i.e. $s_0=s_1=s_2=s_3$) are possible for the ground-state configuration. Note that this is a better bound that the one \eqref{Chois bound} given by
Corollary \ref{Choi thm 2}.

\section{Best bound on the example}
\label{Best bound on the example}
Here we show how to derive the best bound on the internal coupling strength using the toy model of Figure \ref{minorexmaple1} as an example. By Remark \ref{rmk of thm 2}, we have that the best bound is given by
\begin{equation*}
M_i=\min_{h_{i(k)}} M_i\left(h_{i(k)}\right)\,.
\end{equation*}
\begin{figure}[h]
	\centering
	\includegraphics[scale=0.3]{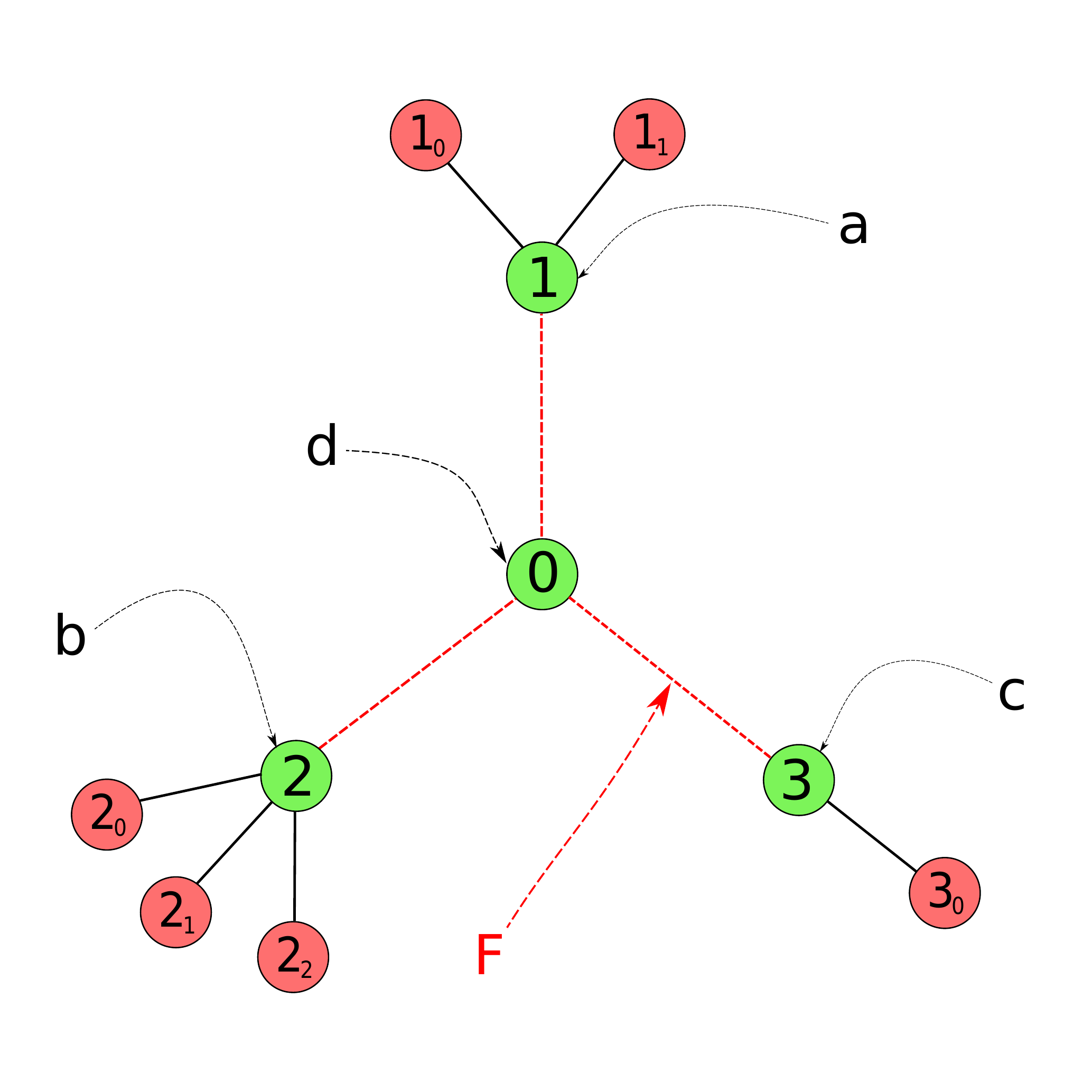}
	\caption{Example \ref{An example} with unspecified $h_{i(k)}$}
	\label{minorexmaple3}
\end{figure}
Now let $h_{i(k)}=\{a,b,c,d\}\,$ and we have the example as shown in Figure \ref{minorexmaple3}. Now follow the same method as in Appendix \ref{An example for the existence of a better bound}, we conclude from Case 1 and 2 inequalities that
\begin{equation}
\left\{
\begin{aligned}
&F<-(5h+a)\,,\\
&F<-(5h-a)\,.
\end{aligned}
\right.
\end{equation}
Therefore we have $F<-5h$ regardless of what value of $a$ takes. This shows that the best constant is $M_i=5h\,$.

\end{appendices}

\end{document}